\documentclass[11pt,reqno]{amsart}
\usepackage{amsmath,amsthm}
\usepackage{lmodern}
\usepackage{latexsym}
\usepackage{amssymb}
\usepackage{enumerate}
\usepackage{fullpage}

\usepackage[usenames,dvipsnames]{xcolor}

\usepackage[unicode=true,pdfusetitle,
bookmarks=true,bookmarksnumbered=true,bookmarksopen=true,bookmarksopenlevel=1,
breaklinks=false,pdfborder={0 0
  0},backref=false,colorlinks=true,linkcolor=NavyBlue,citecolor=NavyBlue,urlcolor=NavyBlue]
{hyperref}
\pdfpageheight\paperheight
\pdfpagewidth\paperwidth

\usepackage{dsfont}
\usepackage{bbm}

\numberwithin{equation}{section}
\newtheorem{theorem}{Theorem}[section]
\newtheorem{lemma}[theorem]{Lemma}
\newtheorem{proposition}[theorem]{Proposition}

\newtheorem*{remark}{Remark}

\newcommand{\mc}[1]{\mathcal{#1}}

\newcommand{\R}{\mathbb{R}}
\renewcommand{\C}{\mathbb{C}}
\newcommand{\Z}{\mathbb{Z}}
\newcommand{\D}{\mathbb{D}}

\renewcommand{\Im}{\mathop{\mathrm{Im}}}
\renewcommand{\epsilon}{\varepsilon}
\newcommand{\SigmaLE}[1][]{\Sigma_{\Lambda_{#1}}^E}
\newcommand{\pr}[1]{\mathbb{P}\left(#1\right)}
\newcommand{\ex}[2][]{\mathbb{E}_{#1}\left(#2\right)}
\newcommand{\cex}[2]{\mathbb{E}\left(#1 \vert #2\right)}
\DeclareMathOperator{\var}{Var}

\DeclareMathOperator{\diag}{diag}
\DeclareMathOperator*{\essinf}{ess\,inf}
\DeclareMathOperator{\mes}{mes}
\DeclareMathOperator{\sgn}{sgn}
\newcommand{\norm}[1]{\lVert#1\rVert}

\DeclareRobustCommand{\SkipTocEntry}[5]{}

\begin{document}

\title[On fluctuations and localization length for the Anderson model]{On fluctuations and
localization length for the Anderson model on a strip}

\author{Ilia Binder}
\address{Dept. of Mathematics, University of Toronto, Toronto, ON, M5S 2E4, Canada}
\email{ilia@math.utoronto.ca}

\author{Michael Goldstein}
\address{Dept. of Mathematics, University of Toronto, Toronto, ON, M5S 2E4, Canada}
\email{gold@math.utoronto.ca}

\author{Mircea Voda}
\address{Dept. of Mathematics, University of Toronto, Toronto, ON, M5S 2E4, Canada}
\email{mvoda@math.utoronto.ca}

\date{}

\begin{abstract}
    We consider the Anderson model on a strip. Assuming that
    potentials have bounded density with considerable tails we get a lower bound for
    the fluctuations of the logarithm of the  Green's function in a finite box.
    This  implies an effective estimate by $ \exp(CW^2) $ for the localization length
    of the Anderson model on the strip of width $ W $. The results are obtained, actually, for
    a more general model with a non-local operator in the vertical direction.
\end{abstract}

\maketitle

\tableofcontents

\section{Introduction}

We consider random operators on the strip $ \Z_W=\Z\times\{1,\ldots,W\} $ defined by 
\begin{equation*}
(H\psi)_n = -\psi_{n-1}-\psi_{n+1}+S_n\psi_n, 
\end{equation*}
where $ \psi\in l^2(\Z,\C^W)\equiv l^2(\Z_W) $, $ S_n=S+\diag(V_{(n,1)},\ldots,V_{(n,W)}) $, with $ 
S $ a Hermitian matrix and $ V_i $, $ i\in\Z_W $, i.i.d. random variables. We assume that $ V_i $ 
have bounded density function $ v $ and we let 
\begin{equation}
\label{eq:intro-densityub} A_0 := \sup_{x} v(x) < + \infty. 
\end{equation}
Furthermore we assume that 
\begin{equation}
\label{eq:intro-densitydecay} \pr{|V_i|\ge T} \le A_1/T, 
\end{equation}
for $ T\ge 1 $.

The problem of estimating the localization length for this model and for the random band matrix 
model is well-known. In the latter case a polynomial bound was established by Schenker 
\cite{Sch-09-Eigenvector}. Very recently, Bourgain \cite{Bou-13-lower} established a bound by $ 
\exp(CW(\log W)^4) $ for the Anderson model, provided that the potentials $ V_i $ have bounded 
density. We will obtain an explicit estimate for the localization length by a method different from 
\cite{Bou-13-lower}. Our approach is via explicit lower bounds for the fluctuations of the Green's 
function. This idea has been previously used by Schenker \cite{Sch-09-Eigenvector}, but our 
implementation is different.

We introduce some notation needed to state our results. Let $\Lambda\subset \Z_W $. For $ 
\Lambda_0\subset\Lambda $ we let $ \Lambda_0'=\Lambda\setminus \Lambda_0 $ and we use $ 
\partial_{\Lambda} \Lambda_0 $ to denote the boundary of $ \Lambda_0 $ relative to $ \Lambda $, 
which is the set of pairs $ (i,i') $ such that $ i\in\Lambda_0 $, $ i'\in \Lambda_0' $, and $ 
|i-i'|=1 $, where $ |j|=\max(|j_1|,|j_2|) $. If $ \Lambda=\Z_W $ we will just write $ 
\partial \Lambda_0 $. If $ (i,i')\in 
\partial_\Lambda \Lambda_0 $ we may also write $ i\in 
\partial_\Lambda \Lambda_0 $ and $ i'\in 
\partial_\Lambda \Lambda_0 $. By $P_\Lambda$ we denote the orthogonal projection onto the subspace 
of all vectors in $ \C^{\Lambda} $ vanishing off $\Lambda$. The restriction of $ H $ to $\Lambda$ 
with Dirichlet boundary conditions is the operator $H_\Lambda: \C^{\Lambda} \to \C^{\Lambda}$, 
defined by $ H_\Lambda:=P_\Lambda H P_\Lambda $. For $ E\subset\Z $ we use $ E_W $ do denote $ 
E\times\{1,\ldots,W\} $. We will use $ \Lambda_L(a) $ to denote $ [a-L,a+L]_W $. Finally, let 
\begin{equation*}
\SigmaLE:=\sum_{i,j\in 
\partial \Lambda,i_1<j_1}|G_{\Lambda}^E(i,j)|^2, 
\end{equation*}
where $ G_\Lambda^E=(H_\Lambda-E)^{-1} $. Note that for $ \Lambda=[a,b]_W$ the above sum is 
over $ i\in\{a\}_W $ and $ j\in\{b\}_W $.

Our estimate on the fluctuations of the resolvent, which will be proved in section 
\ref{sec:analysis}, is as follows.
\begin{theorem}
\label{thm:intro-var-log-Green} There exist  constants $C_0,  C_1=C_1(A_0,|E|,\norm{S}) $ such that
for any $ \Lambda=[a,b]_W $ we have 
\begin{equation*}
	\var(\log \SigmaLE) \ge (b-a-1) (\inf\nolimits_I v)^W, 
\end{equation*}
where $ I=[\pm \exp(CK),\pm \exp((C+C_0)K)] $ , with $ C\ge C_1 $. 
\end{theorem}

The above estimate would work with $ G_\Lambda^E(i,j) $, $ i\in \{ a \}_W $, $ j\in \{ b \}_W $, 
instead of $ \SigmaLE $, but we need the result as is to be able to deduce exponential decay. 
Indeed, employing standard multi-scale analysis, as in \cite{vDK-89-new}, we show in 
Theorem~\ref{thm:decay} that if $ \var(\SigmaLE)\ge (b-a+1)\delta_0 $, $ \delta_0=\delta_0(W) $,
then the localization length is roughly $ \delta_0^{-C} $. Thus, in principle, estimating the 
fluctuations of $ \SigmaLE $ can lead to polynomial bounds on the localization length. In this 
paper we only manage to obtain exponential bounds on the localization length. Concretely, 
Theorem~\ref{thm:intro-var-log-Green} and Theorem~\ref{thm:decay} imply the following estimate on 
the off-diagonal decay of Green's function.
\begin{theorem}
\label{thm:intro-concrete-decay} Fix $ B>0 $ and $ \beta\ge 1 $. There exists  a constant 
\begin{equation*}
	C_0=C_0(A_0,A_1,B,\beta,|E|,\norm{S})
\end{equation*}
 such that if $ 
\inf\nolimits_I v\ge \exp(-BW) $ for some $ I $ as in Theorem~\ref{thm:intro-var-log-Green} then 
\begin{equation*}
	\pr{\log|G_{\Lambda_L(a)}^E(i,j)|\le-\exp(-C_0 W^2)L,\,i\in\{a\}_W, j\in
	\partial \Lambda_L(a)} \ge 1-L^{-\beta}, 
\end{equation*}
for any $ L\ge \exp(2C_0 W^2) $ and $ a\in\Z $. 
\end{theorem}
\begin{remark}
It is well-known, and otherwise straightforward to deduce, that the above estimate implies 
exponential decay of the extended eigenvectors of $ H $, and a lower bound on the non-negative 
Lyapunov exponents. Namely, we have that if $ \gamma_W^E $ is the lowest non-negative Lyapunov 
exponent then $ \gamma_W^E\ge \exp(-CW^2) $, and if $ \psi $ is an extended eigenvector of $ H $ 
then
\begin{equation*}
	\limsup_{|i|\to\infty}(\log|\psi(i)|)/|i|\le -\exp(-CW^2).
\end{equation*} 
\end{remark}

Let us discuss some of the ideas behind the proof of Theorem~\ref{thm:intro-var-log-Green}. The 
strategy is to take advantage of the fact that $ G_\Lambda^E(i,j) $ is the ratio of two polynomials 
of different degrees in $ (V_i)_{i\in\Lambda} $. We illustrate this idea in a simpler setting. If $ 
P(x),Q(x)  $ are two monic polynomials of one 
variable then $ \log|P(x)/Q(x)|\simeq (\deg P-\deg Q )\log |x| $, provided $ |x| $ is large enough. 
If $ \deg P\neq \deg Q $ and large values of $ |x| $ are taken with non-zero probability then the 
previous remark should be enough to capture some of the fluctuations of $ \log|P(x)/Q(x)| $.

The above idea is not sufficient to generate the crucial $ (b-a-1) $ factor in the lower bound on 
variance. Let $ \{\Lambda_k\} $ be a partition of $ \Lambda $ and let 
\begin{equation*}
	h_k(V)=\ex{\log|G_\Lambda^E(i,j)(V,\cdot)|},~V\in \R^{\Lambda_k}
\end{equation*}
(we keep the potentials on $ \Lambda_k $ fixed and we average the rest). Then we have the following 
Bessel type inequality (see 
Lemma~\ref{lem:var-elementary-estimates} (ii)):
\begin{equation*}
	\var(\log |G_\Lambda^E(i,j)|)\ge \sum_k \var(h_k).
\end{equation*}
So, the problem is reduced to estimating the fluctuations of $ h_k $. We get the $ (b-a-1) $ factor 
by just choosing a fine enough partition. Ideally we would choose $ \Lambda_k=\{ k \} $, but this 
turns out to be incompatible with our first idea. Using hyper-spherical coordinates we can write 
$ G_\Lambda^E(i,j)(V,V')=G_\Lambda^E(i,j)(r,\xi,V') $, $ V\in \R^{\Lambda_k} $, $ V'\in 
\R^{\Lambda_k'} $,
$ r\in \R $, $ \xi\in S^{|\Lambda_k|-1} $. Let $ d_1,d_2 $ be the degrees of the numerator and 
denominator of $ G_\Lambda^E(r,\xi,V') $ as polynomials in $ r $. It is then not hard to see that 
the problem of finding a lower bound for $ \var(h_k) $ can be reduced to the problem of estimating 
the variance of a function of the form
\begin{equation*}
	d_1\int_{\C}\log|r-\zeta|\,d\mu_1(\zeta)-d_2\int_{\C}\log|r-\zeta|\,d\mu_2(\zeta),
\end{equation*}
where $ \mu_1,\mu_2 $ are probability measures. Note that if we would have $ \mu_i(|\zeta|\ge R)=0 
$, $ i=1,2 $ then the above function is approximately $ (d_1-d_2)\log r $, for $ r\gg R $, which 
leads us back to our first idea. Clearly, we want $ d_1\neq d_2 $. This is false for 
$ \Lambda_k=\{ k \}, k\in \Lambda $, but it turns out to be true for $ \Lambda_k=\{ k 
\}_W,k\in(a,b) $. The conditions $ \mu_i(|\zeta|\ge R)=0  $, $ i=1,2 $ turn out be roughly 
equivalent to the polynomials on the top and bottom of $ G_\Lambda^E(i,j)(V,V') $ not vanishing for 
$ V 
$ outside the ball of radius $ R $ in $ \C^{\Lambda_k} $ and all $ V'\in \R^{\Lambda_k'} $. 
Unfortunately we can establish such a property only for the denominator of  $ G_\Lambda^E(i,j) $ 
(see Proposition~\ref{prop:analysis-determinant}). This is because the denominator is the 
determinant of a self-adjoint matrix, but the numerator is the determinant of a non-self-adjoint 
matrix. We circumvent this problem at the cost of a worse lower bound on variance. At a technical 
level this is a accounted for by the difference between statements (iii) and (v) of 
Proposition~\ref{prop:var-estimates-for-log-potential}.

Finally, the  ideas discussed above are synthesized in the following theorem, which will be proved 
in 
section \ref{sec:var}. If $ P $ is a polynomial of $ N $ variables and $ J\subset\{1,\ldots,N\} $ 
then $ \deg_J P $ denotes the cumulative degree of $ P $ with respect to the variables indexed by $ 
J $. We will use $ J' $ to denote $ \{1,\ldots,N\}\setminus J $. By $ (x,x') $, $ x\in\R^{J} $, $ 
x'\in\R^{J'} $ we denote the vector in $ \R^{J\cup J'} $ with the components indexed by $ J $ given 
by $ x $ and the components indexed by $ J' $ given by $ x' $.
\begin{theorem}
\label{thm:intro-var-log-rational} Let $ P $ and $ Q $ be two polynomials of $ N $ variables. 
Assume that the following conditions hold: 
\begin{enumerate}
	[(a)] 
	\item There exist $J_k\subset \{1,\dots,N\}$, $ k=1,\dots, N' $, $J_k \cap J_{k'}=\emptyset$ for 
	$k\neq k'$, $ |J_k|=K $ such that 
	\begin{equation*}
		0\le\deg_{J_k}P < \deg_{J_k}Q=K. 
	\end{equation*}
	
	\item For each $k$ and each $T\gg 1$ there exists $\mc B(k,T)\subset \R^{J'_k}$ with 
	$\pr{\mc B(k,T)}\le B_0 K^2 T^{-1}$, such that for any $x'\in \R^{J'_k}\setminus\mc B(k,T)$ and 
	any 
	$x\in\C^{J_k}$ with $\min_i |x_i|\ge T$ we have $ Q(x,x')\neq 0 $. 
\end{enumerate}
Then there exist $ C_0 $, $ C_1=C_1(D) $ such that 
\begin{equation*}
	\var(\log(|P|/|Q|))\ge N' (\inf\nolimits_I v)^K, 
\end{equation*}
for any $ I=[\pm\exp(CK),\pm\exp((C+C_0)K)] $, with $ C\ge C_1 $. 
\end{theorem}

\addtocontents{toc}{\SkipTocEntry}
\subsection*{Acknowledgements}
The authors are grateful to the anonymous referee for his helpful comments. 

\section{Lower bound for the variance of the logarithm of a rational function of several 
variables}\label{sec:var}

In this section we will prove Theorem~\ref{thm:intro-var-log-rational}. The main idea for the proof 
is to reduce the analysis of the variance to the case of a one dimensional logarithmic potential 
for which we have the estimates from Proposition~\ref{prop:var-estimates-for-log-potential}. But 
first we collect some elementary facts concerning the variance. We leave the proofs as an exercise 
for the reader.
\begin{lemma}
\label{lem:var-elementary-estimates} Let $ (\Omega,\mc F,\mu) $ be a probability space. 
\begin{enumerate}
	[(i)] 
	\item If $ X $, $ Y $ are square summable random variables then 
	\begin{equation}
		\label{eq:var-triangle} |\var^{1/2}(X)-\var^{1/2}(Y)|\le \var^{1/2}(X\pm Y) 
	\end{equation}
	and 
	\begin{equation}
		|\var(X)-\var(Y)|\le \ex{(X-Y)^2}^{1/2} \left(\ex{X^2}^{1/2}+\ex{Y^2}^{1/2}\right). 
	\end{equation}
	
	\item If $ X $ is a square summable random variable and $ \mc F_i $, $ i=1,\ldots,n $ are 
	pairwise independent $\sigma $-subalgebras of $ \mc F $ then 
	\begin{equation}
		\label{eq:var-var>Scvar} \var(X)\ge \sum_{i=1}^{n}\var(\cex{X}{\mc F_i}). 
	\end{equation}
	
	\item If $ X $ is a square summable random variable and $ \mu_0 $ is a probability measure such 
	that $ \mu \ge c\mu_0 $, with $ c\in(0,1) $, then 
	\begin{equation}
		\label{eq:var-var>cvar} \var(X)\ge c \var_{\mu_0}(X). 
	\end{equation}
	
	\item If $ \mu_i $, $ i=1,\ldots,n $ are probability measures and $ X_j $, $ j=1,\ldots,m $ are 
	square summable random variables then 
	\begin{equation}
		\label{eq:var-sumvar} \sum_i \var_{\mu_i}\bigg(\sum_j \beta_j X_j\bigg) \le \bigg(\sum_j 
		|\beta_j|\bigg)^2 \max_j \sum_i \var_{\mu_i}(X_j). 
	\end{equation}
	
	\item If $ (\Omega',\mc F', \mu') $ is a probability space and $ X $ is a square summable random 
	variable on $ \Omega\times \Omega' $ then 
	\begin{equation}
		\label{eq:var-infvar} \var_{\mu\times\mu'}(X) \ge \essinf_{\omega'\in\Omega'} 
		\var_{\mu}(X(\cdot,\omega')). 
	\end{equation}
\end{enumerate}
\end{lemma}

From now on we will reserve $ d\nu $ for the joint probability distribution of $ 
(V_i)_{i\in\Lambda} $, where $ \Lambda $ will be clear from the context. We use $ dm_\Omega $ for 
the uniform distribution on $ \Omega\subset \R^d $ (with $ d $ clear from the context) and $ 
\var_\Omega(\cdot) $, $ \ex[\Omega]{\cdot} $ will be computed with respect to $ dm_\Omega $. The 
statement of the next result exposes the main steps of its proof. We note that the statements 
relevant for the proof of Theorem~\ref{thm:intro-var-log-rational} are (iii) and (v).

\begin{proposition}
\label{prop:var-estimates-for-log-potential} Let $\mu$ be a Borel probability measure on $\C$ and 
set 
\begin{equation*}
	u_\mu(x):=\int_{\C}\log |x-\zeta|d\mu(\zeta). 
\end{equation*}
We assume that $ \mu $ is such that $ u_\mu $ is locally square summable.
\begin{enumerate}
	[(i)]
	
	\item If $\mu ( \{|\zeta|\ge R\})=0$ for some $R>0$, then for any $ M>0 $ one has 
	\begin{equation*}
		\ex[{[0,M]}]{u_{\mu}^2}\le \frac{4\min(1,M)(\log(\min(1,M))-1)^2+M\log^2(M+R)}{M}. 
	\end{equation*}
	
	\item $ \var_{[M_0,M_1]}(u_{\mu})=\var_{[M_0M_1^{-1},1]}(u_{\mu^{(M_1)}})$, for any $M_1>M_0\ge 
	0$, where $\mu^{(M_1)}(\cdot):=\mu(M_1\cdot)$.
	
	\item If $\mu ( \{|\zeta|\ge R\})=0$ for some $R>0$, then for any $M_1\ge 2M_0\ge 4R$ one has 
	\begin{equation*}
		\left|\var_{[M_0,M_1]}(u_\mu) -1\right| \le 10^4 
		\left((RM_1^{-1})^{1/5}+(M_0M_1^{-1})^{1/2}\right). 
	\end{equation*}
	
	\item If $\mu (\{|\zeta|\le R\})=0$ for some $R>0$, then for any $ 0\le 2M_0\le M_1\le R/2$ one 
	has
	\[ \var_{[M_0,M_1]}(u_\mu)\le 8(M_1R^{-1})^2. \]
	
	\item For any $ M_0\ge0 $ we have 
	\begin{equation*}
		\sum_{k=1}^{m}\var_{[M_0,M_k]}(u_{\mu})<m+10^5, 
	\end{equation*}
	with $ M_k=2^kA_0 $, $ A_0>0 $, $ A_0\ge M_0 $. In particular, for any $ m\ge 1 $, there exists $ 
	M\in[2A_0,2^m A_0] $ such that $ \var_{[M_0,M]}(u_{\mu})<1+10^5 m^{-1} $. 
\end{enumerate}
\end{proposition}

\begin{proof}
Recall that for $ A>0 $ we have 
\begin{equation*}
	\int_{0}^{A}\log x \,dx=A(\log A -1),\quad \int_{0}^{A}\log^2 x\,dx=A[(\log A -1)^2+1]. 
\end{equation*}
(i) 
\begin{multline*}
	\ex[{[0,M]}]{u_{\mu}^2} 
	\le \frac{1}{M}\int_{0}^{M}\int_{|\zeta|\le R}(\log|x-\zeta|)^2\,d\mu(\zeta)\,dx\\
	=\frac{1}{M}\int_{|\zeta|\le R}\Bigg( \int_{x\in[0,M],|x-\zeta|<1}(\log|x-\zeta|)^2\,dx\\ 
	+\int_{x\in[0,M],|x-\zeta|\ge 1}(\log|x-\zeta|)^2\,dx\Bigg)\,d\mu(\zeta)\\
	\le\frac{1}{M}\int_{|\zeta|\le R}\left( 2 \int_{0}^{\min(1,M)}(\log y)^2\,dy 
	+M(\log(M+R))^2\right)\,d\mu(\zeta)\\
	\le\frac{4\min(1,M)(\log(\min(1,M))-1)^2+M\log^2(M+R)}{M}. 
\end{multline*}

(ii) By a change of variables we have 
\begin{equation*}
	\var_{[M_0,M_1]}(u_{\mu})=\var_{[M_0M_1^{-1},1]}(u_{\mu}(M_1\cdot)).
\end{equation*}
 Now the conclusion follows 
from the fact that 
\begin{equation*}
	u_\mu(M_1x)=u_{\mu^{(M_1)}}(x)+ \log M_1.
\end{equation*}

(iii) First note that 
\begin{equation}
	\label{eq:log-mvt} |\log|x-\zeta|-\log|x||\le 2 |x|^{-1}|\zeta|,~|x|^{-1}|\zeta|\le1/2, 
\end{equation}
and consequently 
\begin{equation*}
	|u_{\mu^{(M_1)}}(x)-\log x|\le 2 \sqrt{RM_1^{-1}},~x\in[\sqrt{RM_1^{-1}},1]. 
\end{equation*}
By what we already established we have 
\begin{multline*}
	|\var_{[M_0,M_1]}(u_{\mu})-\var_{[M_0,M_1]}(\log)| \\
	=|\var_{[M_0M_1^{-1},1]}(u_{\mu^{(M_1)}})-\var_{[M_0M_1^{-1},1]}(\log)|\\
	\le \norm{u_{\mu^{(M_1)}}-\log}_{L^2_{[M_0M_1^{-1},1]}} 
	\left(\norm{u_{\mu^{(M_1)}}}_{L^2_{[M_0M_1^{-1},1]}} +\norm{\log}_{L^2_{[M_0M_1^{-1},1]}}\right)\\
	\le 2\norm{u_{\mu^{(M_1)}}-\log}_{L^2_{[0,1]}} \underbrace{ 
	\left(\norm{u_{\mu^{(M_1)}}}_{L^2_{[0,1]}} +\norm{\log}_{L^2_{[0,1]}}\right) }_{< 5}\\
	\le 10\left(4RM_1^{-1} +\int_{0}^{\sqrt{RM_1^{-1}}}2(u^2_{\mu^{(M_1)}}(x)+\log^2 
	x)\,dx\right)^{1/2}\\
	\le 10\left(4RM_1^{-1}+350\sqrt{RM_1^{-1}}\log^2\sqrt{RM_1^{-1}}\right)^{1/2}\\
	\le 100 (RM_1^{-1})^{1/4}\log(M_1R^{-1}) \le2000 (RM_1^{-1})^{1/5}. 
\end{multline*}
Now we just have to estimate $ \var_{[M_0,M_1]}(\log)=\var_{[m,1]}(\log) $, where we let $ 
m=M_0M_1^{-1} $. 
\begin{multline}
	|\var_{[m,1]}(\log)-1|
	=\Bigg|\ex[{[m,1]}]{\log^2}-\frac{1}{1-m}\ex[{[0,1]}]{\log^2} 
	\\-\left(\ex[{[m,1]}]{\log}\right)^2 
	+\frac{1}{(1-m)^2}\left(\ex[{[0,1]}]{\log}\right)^2 -\frac{m^2}{(1-m)^2}\Bigg|\\
	\le \frac{1}{1-m}\left|\int_{0}^{m}(\log x)^2\,dx\right| +\frac{1}{(1-m)^2}\left|\int_{0}^{m}\log 
	x\,dx\right| \left|\int_{0}^{1}\log x\,dx+\int_{m}^{1}\log x\,dx\right| \\+\frac{m^2}{(1-m)^2}
	\le \frac{5m(1-\log m)^2}{(1-m)^2}\le 500 m \log^2 m \le 10^4 m^{1/2}. 
\end{multline}

(iv) Note that based on \eqref{eq:log-mvt} we have 
\begin{equation*}
	|u_{\mu^{(M_1)}}(x)-u_{\mu^{(M_1)}}(0)|\le 2 M_1R^{-1},~x\in[M_0 M_1^{-1},1], 
\end{equation*}
and hence 
\begin{multline*}
	\var_{[M_0,M_1]}(u_{\mu}) =\var_{[M_0 M_1^{-1},1]}(u_{\mu^{(M_1)}}) \le 
	\norm{u_{\mu^{(M_1)}}-u_{\mu^{(M_1)}}(0)}^2_{L^2_{[M_0 M_1^{-1},1]}}\\
	\le \frac{4(M_1R^{-1})^2}{1-M_0M_1^{-1}}\le 8(M_1R^{-1})^2. 
\end{multline*}

(v) Let $ D_l=\{M_l\le |\zeta| < M_{l+1}\} $, $ l=1,\ldots,m-1 $, $ D_{0}=\{|\zeta|<M_1\} $, and $ 
D_{m}=\{|\zeta|\ge M_m\} $. We have $ u_{\mu}=\sum_{l=0}^{m}\mu(D_l)u_{\mu_{D_l}} $, where $ 
\mu_{D}=\mu(D)^{-1}\mu\vert_{D} $ (we set $ \mu_D=0 $ if $ \mu(D)=0 $). We will verify the estimate 
in (v) for each measure $ \mu_{D_l} $. The estimate for $ \mu $ will follow by 
\eqref{eq:var-sumvar}. So, fix arbitrary $l\in\{0,\ldots,m\}$. One has due to part (iv) that 
\begin{equation*}
	\sum_{k=1}^{l-1}\var_{[M_0,M_k]} (u_{\mu_{D_l}}) \le \sum_{k=1}^{l-1}8(M_kM_{l}^{-1})^2 
	=8\sum_{k=1}^{l-1}4^{k-l} \le 8. 
\end{equation*}
On the other hand due to part (iii) one has 
\begin{multline*}
	\sum_{k=l+3}^{m}\var_{[M_0,M_k]} (u_{\mu_{D_l}}) \le \sum_{k=l+3}^{m}\left[1 
	+10^4\left((M_0M_k^{-1})^{1/2} +(M_{l+1}M_k^{-1})^{1/5}\right) \right]\\
	\le m+10^4\left(\sum_{k=1}^{\infty}2^{-k/2}+\sum_{k=1}^{\infty}2^{-k/5}\right) \le m+5\cdot 10^4. 
\end{multline*}
Now we just have to evaluate the variance for $ l\le k\le l+2 $. For $ l<m $ we use (i) to get 
\begin{multline*}
	\sum_{k=l}^{l+2}\var_{[M_0,M_k]}(u_{\mu_{D_l}}) 
	=\sum_{k=l}^{l+2}\var_{[M_0M_k^{-1},1]}(u_{\mu_{D_l}}^{(M_k)}) \\ \le\sum_{k=l}^{l+2} 
	\frac{1}{1-M_0M_k^{-1}}\norm{u_{\mu_{D_l}}^{(M_k)}}_{L^2_{[0,1]}}^2
	\le 2\sum_{k=l}^{l+2}\left(4+\log^2(1+M_{l+1}M_k^{-1})\right) \le 40. 
\end{multline*}
When $ l=m $ we just need to evaluate $ \var_{[M_0,M_m]}(u_{\mu_{D_m}}) $. Let $ 
D_m^1=\{M_m\le|\zeta|< 2M_m\} $ and $ D_m^2=\{|\zeta|\ge 2M_m\} $. Using \eqref{eq:var-sumvar}, (i) 
(for $ u_{\mu_{D_m^1}} $, as above), and (iv) (for $ u_{\mu_{D_m^2}} $) we get 
\begin{multline*}
	\var_{[M_0,M_m]}(u_{\mu_{D_m}}) \le \max \left( 
	\var_{[M_0,M_m]}(u_{\mu_{D_m^1}}),\var_{[M_0,M_m]}(u_{\mu_{D_m^2}}) \right)\\
	\le \max\left( 4+\log^2(1+2M_m/M_m),8(M_m/(2M_m))^2 \right)\le 10. 
\end{multline*}
This concludes the proof.
\end{proof}

Before we proceed with the proof of Theorem~\ref{thm:intro-var-log-rational} we need the two 
following auxiliary results.
\begin{lemma}
\label{lem:var-logP-2nd-moment} If $ P(x)=\sum_{|\alpha|\le D} a_\alpha x^\alpha $ is a polynomial 
of $ N $ variables such that $ \max_{|\alpha|\le D} |a_\alpha|=1$, and $ \Omega\subset \{x\in 
\R^N:~\norm{x}\le R_0 \}$, $ R_0\ge e $, is such that $ \mes(\Omega)>1 $, then 
\begin{equation*}
	\ex[\Omega]{\log^2|P|}\lesssim D^2 N^2 \log^2(N+1) \log^4 R_0. 
\end{equation*}
\end{lemma}
\begin{proof}
The polynomial $ P $ has at most $ (N+1)^D $ monomials, so for $ R\ge e $ we have 
\begin{equation*}
	\sup_{\norm{z}\le R}\log|P(z)|\le \log(R^D(N+1)^D)\lesssim D\log(N+1)\log R. 
\end{equation*}
Lemma~\ref{lem:apx-logP-Cartan} implies that 
\begin{equation*}
	\mes\{x\in\R^N:~\norm{x}\le R,~ \log|P(x)|\le-CHD\log(N+1)\log(20R)\} \le C^N R^N\exp(-H), 
\end{equation*}
for $ H\gg 1 $. The conclusion follows from Lemma~\ref{lem:apx-Cartan->integrability}. 
\end{proof}
\begin{lemma}
\label{lem:var-sectorarea} Let $ \sigma $ be the spherical measure on the $ (n-1) $-sphere $ 
S^{n-1} $. 
\begin{equation*}
	\sigma(\{\xi\in S^{n-1}:~\min_i |\xi_i|\ge \epsilon\}) \ge n 2^n (1-\sqrt{n}\epsilon)^n. 
\end{equation*}
\end{lemma}
\begin{proof}
Let $ \Theta $ be the set whose measure we want to estimate and let 
\begin{equation*}
	\Omega=\{x\in \R^n:~1\le \min_i |x_i|,~\max_i |x_i|\le 1/(\sqrt{n}\epsilon)\}. 
\end{equation*}
Then we have 
\begin{equation*}
	\Omega\subset\{r\xi:~\xi\in\Theta,~r\in[1,1/(\sqrt{n}\epsilon)]\}, 
\end{equation*}
and the conclusion follows from 
\begin{equation*}
	2^n \left(\frac{1}{\sqrt{n}\epsilon}-1\right)^n =\mes(\Omega)\le 
	\int_{\Theta}\int_{1}^{1/(\sqrt{n}\epsilon)}r^{n-1}drd\sigma(\xi) \le 
	\frac{1}{n}\left(\frac{1}{\sqrt{n}\epsilon}\right)^n \sigma(\Theta). 
\end{equation*}
\end{proof}
\begin{proof}
[Proof of Theorem~\ref{thm:intro-var-log-rational}] Set $ h(x):= \log(|P(x)|/|Q(x)|) $. Due to 
\eqref{eq:var-var>Scvar} one has 
\begin{equation}\label{eq:var-Bessel-h}
	\var(h)\ge \sum_k \var(\cex{h}{\mc J_k})=\sum_k\var(h_k), 
\end{equation}
where $ \mc J_k $ is the $ \sigma $-algebra corresponding to fixing the components with indices in 
$ J_k $, and $ h_k(x)=\ex{h(x,\cdot)} $, $ x\in \R^{J_k} $.

To provide a lower bound for $ \var(h_k) $ we will pass to a uniform distribution and we will use 
hyper-spherical coordinates to pass to a one-dimensional problem. Let $ I=[M_0/(2\sqrt{K}),M] 
$, with $ 
M=2^{10^6}K M_0$, $ M_0=2\sqrt{K}T $, $ T=B_0\exp(CK) $, $ C\gg 10^6 $. We define 
\begin{equation*}
	\Theta=\{\xi\in S^{K-1}:~\min_i \xi_i\ge 1/(2\sqrt{K})\}
\end{equation*}
and 
\begin{equation*}
	\Omega=\{x\in \R^K: x=r\xi,~r\in[M_0,M],~\xi\in\Theta\}.
\end{equation*}
The peculiar choice of $ \Theta $ is so that we will be able to use the assumptions on $ Q $. Note 
that for $ x\in \Omega $ we have $ x_i\in I $. Furthermore, by Lemma~\ref{lem:var-sectorarea} we 
have $ \sigma(\Theta)\ge K2^{-K} $ and consequently 
$ \mes(\Omega)\ge 2^{-K}(M^K-M_0^K) $.  By \eqref{eq:var-var>cvar}  we have 
\begin{equation*}
	\var(h_k)\ge (\inf\nolimits_{I} v)^K \mes(\Omega)\var_{\Omega}(h_k) \ge (\inf\nolimits_I 
	v)^K2^{-K}(M^K-M_0^K)\var_{\Omega}(h_k). 
\end{equation*}
Changing variables to hyper-spherical coordinates we have 
\begin{equation*}
	\var_{\Omega}(h_k)=\var_{\eta}(h_k),
\end{equation*}
where 
\begin{equation*}
	d\eta:=(Kr^{K-1}dr/(M^K-M_0^K))\times (d\sigma/\sigma(\Theta)) 
\end{equation*}
is the probability measure on $ \mc R=[M_0,M]\times\Theta $. Using \eqref{eq:var-var>cvar} we can 
pass to the uniform distribution on $ \mc R $: 
\begin{equation*}
	\var_\Omega(h_k)\ge K (M-M_0)M_0^{K-1}/(M^K-M_0^K)\var_{\mc R}(h_k). 
\end{equation*}
Finally, due to \eqref{eq:var-infvar} we have 
\begin{equation*}
	\var_{\mc R}(h_k)\ge \essinf_{\xi\in\Theta}\var_{[M_0,M]}(h_k(\cdot,\xi)), 
\end{equation*}
where $ h_k(r,\xi)=h_k(r\xi) $. In conclusion we have 
\begin{equation}
	\label{eq:var-lb-essinf-r} \var(h_k)\ge K(M-M_0)M_0^{K-1} 2^{-K} (\inf\nolimits_I v)^K 
	\essinf_{\xi\in\Theta}\var_{[M_0,M]}(h_k(\cdot,\xi)). 
\end{equation}

To be able to use the assumption on $ Q $ we want to work with a truncated version of $ h_k $ 
obtained by averaging only on $ \mc G_k:=\R^{J'_k}\setminus\mc B(k,T) $, Passing from the variance 
of $ h_k $ to the variance of the truncated function will depend on having an explicit bound on the 
second moment of $ h_k $. The bound will follow using Lemma~\ref{lem:var-logP-2nd-moment} after an 
appropriate normalization. We know $ P $ and $ Q $ are polynomials in $ r $ and we can write $ 
P(r,\xi,x')=\sum_i a_i(\xi,x') r^i $, $ Q(x)=\sum_i b_i(\xi,x') r^i $. Let $ A(\xi,x')=\max_i 
|a_i(\xi,x')|$, $ B(\xi,x')=\max_i |b_i(\xi,x')|$, and define $ 
\hat{P}(r,\xi,x')=P(r,\xi,x')/A(\xi,x') $, $ \hat{Q}(r,\xi,x')=Q(r,\xi,x')/B(\xi,x') $, and 
\begin{equation*}
	\hat{h}=\log|\hat{P}/\hat{Q}|.
\end{equation*} These functions are well-defined for $ \sigma\times\nu $-almost 
all $ (\xi,x') $. From now on we fix $ \xi $ such that the functions are well-defined for $ \nu 
$-almost all $ x' $. Of course, this means $ \xi $ must be outside a set of measure $ 0 $, but this 
doesn't affect the essential infimum in \eqref{eq:var-lb-essinf-r}. Since $ 
\ex{|\log|A(\xi,\cdot)||}, 
\ex{|\log|B(\xi,\cdot)||}<\infty $ we have 
\begin{equation*}
	\var_{[M_0,M]}(h_k(\cdot,\xi))=\var_{[M_0,M]}(\hat{h}_k(\cdot,\xi)), 
\end{equation*}
where 
\begin{equation*}
	\hat{h}_k(r,\xi)=h_k(r,\xi)-\ex{\log|A(\xi,\cdot)|}+\ex{\log|B(\xi,\cdot)|}. 
\end{equation*}
Using Lemma~\ref{lem:var-logP-2nd-moment} we obtain 
\begin{multline*}
	\ex[{[M_0,M]}]{\hat{h}^2_k(\cdot,\xi)} =\int_{[M_0,M]} 
	\left(\int_{\R^{J'_k}}\hat{h}(r,\xi,x')\,d\nu(x')\right)^2 \,dm_{[M_0,M]}(r)\\
	\le \int_{\R^{J'_k}} \left(\int_{[M_0,M]}\hat{h}^2(r,\xi,x') \,dm_{[M_0,M]}(r) \right)d\nu(x') 
	\lesssim K^2\log^4 M. 
\end{multline*}
We now introduce the truncated version of $ \hat{h}_k $: 
\begin{equation*}
	\tilde{h}_k(r,\xi)=\int_{\mc G_k}\hat{h}(r,\xi,x')\frac{d\nu(x')}{\pr{\mc G_k}}. 
\end{equation*}
By the same argument as for $ \hat{h}_k(\cdot,\xi) $ we have $ 
\ex[{[M_0,M]}]{\tilde{h}^2_k(\cdot,\xi)}\lesssim K^2\log^4 M $ and 
\begin{equation*}
	\ex[{[M_0,M]}]{(\hat{h}_k(\cdot,\xi)-\pr{\mc G_k}\tilde{h}_k(\cdot,\xi))^2} \lesssim \pr{\mc 
	B(k,T)}K^2\log^4 M. 
\end{equation*}
We now get 
\begin{multline*}
	|\var_{[M_0,M]}(\hat{h}_k(\cdot,\xi))-\var_{[M_0,M]}(\pr{\mc G_k}\tilde{h}_k(\cdot,\xi))|\\
	\le \ex[{[M_0,M]}]{(\hat{h}_k(\cdot,\xi)-\pr{\mc G_k}\tilde{h}_k(\cdot,\xi))^2}^{1/2}
		\left(\ex[{[M_0,M]}]{\hat{h}^2_k(\cdot,\xi)}^{1/2}
			+\ex[{[M_0,M]}]{\tilde{h}^2_k(\cdot,\xi)}^{1/2}\right)\\
	\lesssim \pr{\mc B(k,T)}^{1/2}K^2\log^4 M. 
\end{multline*}
We claim that $ \var_{[M_0,M]}(\tilde{h}_k(\cdot,\xi))\ge 2^{-10^6K} $. Since we chose
\begin{equation*}
	T=B_0\exp(CK),\,C\gg 10^6
\end{equation*}
it follows that
\begin{multline*}
	\label{eq:var-truncatedlb} \var_{[M_0,M]}(h_k(\cdot,\xi))\ge \pr{\mc 
	G_k}^2\var_{[M_0,M]}(\tilde{h}_k(\cdot,\xi)) -C\pr{\mc B(k,T)}^{1/2}K^2\log^4 M\\
	\ge\var_{[M_0,M]}(\tilde{h}_k(\cdot,\xi))/2\ge 2^{-10^6K}/2. 
\end{multline*}
From this, \eqref{eq:var-lb-essinf-r}, and \eqref{eq:var-Bessel-h} it follows that
\begin{equation*}
	\var(h)\ge N'K(M-M_0)M_0^{K-1} 2^{-(K+1)} 2^{-10^6K} (\inf\nolimits_I v)^K.
\end{equation*}
Note that by our choice of $ M_0,M,T $ we have
\begin{equation*}
	K(M-M_0)M_0^{K-1} 2^{-(K+1)} 2^{-10^6K}=\exp(CK^2)\ge 1,
\end{equation*}
so the desired lower bound on variance follows. The case 
\begin{equation*}
	I=[-M_1,-M_0/(2\sqrt{K})]
\end{equation*}
 follows 
analogously. Note that in fact we obtained a better estimate than the one 
stated in the theorem. However, it can be seen that $ (\inf\nolimits_I v)^K\le \exp(-C'K^2) $ with 
$ C'\gg C $, so the estimate won't be substantially better than the stated one. 

Now we just have to show that $ \var_{[M_0,M]}(\tilde{h}_k(\cdot,\xi))\ge 2^{-10^6K} $. Using 
\eqref{eq:var-var>cvar} we get 
\begin{equation}
	\label{eq:var-Mxi} \var_{[M_0,M]}(\tilde{h}_k(\cdot,\xi)) \ge (M_\xi-M_0)/(M-M_0) 
	\var_{[M_0,M_\xi]}(\tilde{h}_k(\cdot,\xi)), 
\end{equation}
with $ M_\xi\in(M_0,M) $ to be chosen later. We provide a lower bound for $ 
\var_{[M_0,M_\xi]}(\tilde{h}_k(\cdot,\xi)) $ by applying 
Proposition~\ref{prop:var-estimates-for-log-potential}. We first need to set-up $ \tilde{h}_k $ as 
the difference of two logarithmic potentials. Without loss of generality we may assume that $ \hat 
P $ and $ \hat Q $ are monic polynomials in $ r $ (we can force them to be so, without changing the 
variance). Let $ D_k $ be the degree in $ r $ of $ \hat{P}(r,\xi_0,x') $. If $ D_k=0 $ then the 
term corresponding to $ \hat P $ won't contribute to the variance. So, we only deal with the case $ 
D_k\ge 1 $. It is well-known that there exist measurable functions $ \zeta_j $ such that 
\begin{equation*}
	\hat{P}(r,\xi,x')=\prod\limits_{j=1}^{D_k}(r-\zeta_j(x')). 
\end{equation*}
Let $ \mu_j $ be the push-forward of the measure $ (\nu\vert_{\mc G_k})/\pr{\mc G_k} $ under the 
map $ x' \to \zeta_j(x') $. Let $ u_k(r)=\int_{\C}\log|r-\zeta|\,d\mu_P(\zeta) $, where $ \mu_P $ 
is the probability measure defined by $ \mu_P=D_k^{-1} \sum_j \mu_j$. Analogously, we define $ 
v_k(r)=\int_{\C}\log|r-\zeta|\,d\mu_Q(\zeta) $ to be the logarithmic potential corresponding to $ 
\hat{Q}(r,\xi_0,x') $. Note that both $ u_k $ and $ v_k $ are square summable, and furthermore by 
the choice of $ \mc G_k $ and $ \Theta $ we have $ \mu_Q(|\zeta|\ge 2\sqrt{K}T)=0 $ (this is 
equivalent to saying that $ \hat{Q}(r,\xi,x')\neq0 $, for $ |r|\ge2\sqrt{K}T $, $ \xi\in \Theta $, 
$ x'\in \mc G_k $, which is true by assumption (ii) of the theorem). We have 
\begin{equation*}
	\tilde{h}_k(r,\xi)=D_ku_k(r)-Kv_k(r). 
\end{equation*}
By part (iii) of Proposition~\ref{prop:var-estimates-for-log-potential} we get 
\begin{equation*}
	\var_{[M_0,M_{\xi}]}(v_k)\ge 1-(4K)^{-1}, 
\end{equation*}
for any $ M_{\xi}\ge 4^5 10^{20}K^5 M_0 $. Using part (v) of 
Proposition~\ref{prop:var-estimates-for-log-potential} we choose
\begin{equation*}
	M_{\xi}\in[2 \cdot4^5 10^{20} K^5 M_0, 2^{4\cdot 10^5K} 4^5 10^{20} K^5 M_0]\subset (M_0,M_1), 
\end{equation*}  such that 
\begin{equation*}
	\var_{[M_0,M_{\xi}]}(u_k)\le 1+(4K)^{-1}. 
\end{equation*}
Using \eqref{eq:var-triangle} we have 
\begin{multline*}
	\var_{[M_0,M_{\xi}]}(\tilde{h}_k(\cdot,\xi)) \ge\left(\var^{1/2}_{[M_0,M_{\xi}]}(D_ku_k) 
	-\var^{1/2}_{[M_0,M_{\xi}]}(Kv_k)\right)^2\\
	\ge\left(K(1-(4K)^{-1})^{1/2}-(K-1)(1+(4K)^{-1})^{1/2}\right)^2 \ge 1/4. 
\end{multline*}
Plugging the above estimate in \eqref{eq:var-Mxi} yields that 
\begin{equation*}
	\var_{[M_0,M]} (\tilde h_k(\cdot,\xi))
	\ge \frac{M_0(2\cdot4^5 10^{20} K^5-1)}{4M_0(2^{10^6K}-1)}\ge 2^{-10^6K}. 
\end{equation*}
This concludes the proof.
\end{proof}

\section{Analysis of the determinant and of the minors as polynomials in terms of the 
potentials}\label{sec:analysis}

Let $ f_\Lambda^E=\det(H_\Lambda-E) $ and let $ g_\Lambda^E(i,j) $ be the $ (i,j) $ minor of $ 
H_\Lambda-E $. In this section we are interested in $ f_\Lambda^E $ and $ g_\Lambda^E(i,j) $ as 
polynomials in $ (V_i)_{i\in\Lambda} $. We will prove Theorem~\ref{thm:intro-var-log-Green}, as a 
consequence of Theorem~\ref{thm:intro-var-log-rational}, and we will provide bounds on the moments 
of $ \SigmaLE $, which will be needed in section \ref{sec:decay}. The properties of $ f_\Lambda^E $ 
and $ g_\Lambda^E(i,j) $ that are needed for these results are established in the next two 
propositions.

In the following it is useful to keep in mind that if we order the points of $ \Z_W $ 
lexicographically, i.e. $ i < j $ if $ i_1 < j_1 $, or $ i_1 = j_1 $ and $ i_2 < j_2 $, then the 
matrix of $H_{\Lambda}$, $ \Lambda=[a,b]_W $, is 
\begin{equation*}
\begin{bmatrix}
	S_a & - I & 0 & 0 &\hdotsfor{2}\\
	-I & S_{a+1}& -I & 0 & \hdotsfor{2}\\
	\ldots & \ddots & \ddots & \ddots &\hdotsfor{2}\\
	\hdotsfor{2} & \ddots & \ddots & \ddots & \ldots\\
	\hdotsfor{2} & 0 & -I & S_{b-1} & -I\\
	\hdotsfor{2} & 0 & 0 & -I & S_{b}\\
\end{bmatrix}
. 
\end{equation*}

	For the application of Theorem~\ref{thm:intro-var-log-rational} we will only need the first 
	part of the following result. The second part will be needed for establishing the Cartan type 
	estimate for $ \log \SigmaLE $ in Lemma~\ref{lem:analysis-Cartan-R}.
\begin{proposition}
\label{prop:analysis-minors} Let $ i,j\in \Lambda=[a,b]_W $ be such that $ i_1<j_1 $ and let $ 
n\in(i_1,j_1) $.
\begin{enumerate}
	[(i)] 
	\item The degree of $ g_\Lambda^E(i,j) $ as a polynomial of $ (V_k)_{k\in \{n\}_W} $ is at most $ 
	W-1 $.
	
	\item If $ i_2=j_2 $ then the polynomial $ [g_\Lambda^E(i,j)](V) $ has a monomial whose 
	coefficient is $ \pm 1 $. Furthermore, the degree of $ [g_\Lambda^E(i,j)](V) $ as a polynomial of 
	$ (V_k)_{k\in \{n\}_W} $ is $ W-1 $. 
\end{enumerate}
\end{proposition}
\begin{proof}
It is enough to prove the result for $ E=0 $.

(i) $ g_\Lambda^E(i,j) $ is the determinant of a matrix of the form 
\begin{equation*}
	\begin{bmatrix}
		* & * & 0 \\
		* & S_n & * \\
		0 & * & * 
	\end{bmatrix}
	, 
\end{equation*}
where the top-right corner entry is a $ (p-1)\times(q-1) $ matrix and the lower-left corner entry 
is a 

$ q\times p $
 matrix, with $ p=(n-a)W $ and $ q=(b-n)W 
$. The coefficient of the monomial $ \prod_{k\in\{n\}_W}V_k $ is (up to sign) the determinant of 
the matrix obtained by removing the rows and and the columns corresponding to $ S_n $. This matrix 
is of the form 
\begin{equation*}
	\begin{bmatrix}
		* & * & 0 \\
		0 & 0 & * \\
		0 & 0 & * 
	\end{bmatrix}
	, 
\end{equation*}
where the entries on the diagonal are blocks of size $ (p-1)\times(p-1) $, $ 1\times1 $, and $ 
(q-1)\times(q-1) $ respectively. Hence the determinant is zero and the conclusion follows.

(ii) For fixed $ i,j\in\Lambda $ let $ H_\Lambda^{ij} $ be the operator corresponding the matrix 
obtained from $ H_\Lambda $ by making all entries on the $ i $-th row and on the $ j $-th column 
zero, except for the $ (i,j) $-th entry which is set to $ 1 $. Up to sign, $ g_\Lambda^E(i,j) $ is 
the determinant of $ H_\Lambda^{ij} $. We will use $ h $ to denote the entries of the matrix 
representation of $ H_\Lambda^{ij} $. By the Leibniz formula for determinants 
\begin{equation*}
	g_\Lambda^E(i,j)=\sum_\sigma \sgn(\sigma)\prod_{l\in\Lambda}h_{l,\sigma(l)}, 
\end{equation*}
where $ \sigma $ runs over all permutations of $ \Lambda $. We are interested in the non-zero terms 
from the above sum that are divisible by $ V^\alpha $ where $ \alpha\in \{0,1\}^\Lambda $ and 
\begin{equation*}
	\alpha_l=
	\begin{cases}
		1 & \text{if } l_1\notin[i_1,j_1]\text{, or } l_1\in[i_1,j_1] \text{ and } l_2\neq j_2\\
		0 & \text{otherwise} 
	\end{cases}
	. 
\end{equation*}
For each $ l $ there are at most $ W+2 $ values for $ \sigma(l) $ such that $ h_{l,\sigma(l)} $ is 
not zero. The permutations $ \sigma $ corresponding to non-zero terms divisible by $ V^\alpha $ 
must satisfy $ \sigma(l)=l $ when $ \alpha_l=1 $. 
	It follows that for such permutations we have $ \sigma([i_1,j_1]\times \{ j_2 
	\})=[i_1,j_1]\times \{ j_2 \} $. Note that by our definition of $ H_\Lambda^{ij} $ we must
	have $ \sigma(i_1,j_2)=(j_1,j_2) $. Hence we must have $ \sigma((i,j_2)) = 
	(i-1,j_2) $, for any $ i\in(i_1,j_1] $. So $ h_{l,\sigma(l)}=\pm 1 $, whenever $ \alpha_l=0 $.
	
This shows that the monomial $ V^\alpha $ has coefficient $ \pm 1 $. From this it also follows that 
the degree of $ [g_\Lambda^E(i,j)](V) $ as a polynomial of $ (V_k)_{k\in\{n\}_W} $ is at least $ 
W-1 $. Now the conclusion follows from part (i). 
\end{proof}
\begin{remark}
The second part of the previous proposition doesn't necessarily hold when $ i_2\neq j_2 $. In 
particular, it can be seen that $ g_\Lambda^E(i,j) $ is identically zero for any $ i,j\in \Lambda 
$, with $ i_2\neq j_2 $, provided that $ S=0 $. 
\end{remark}

For the next result we will need some bounds on the probability distribution of the resolvent. From 
\cite[Theorem II.1]{AM-93-Localization} we have 
\begin{equation}
\label{eq:Wegner-for-entries} \pr{|G_\Lambda^E(i,j)|\ge T}\lesssim A_0/T, 
\end{equation}
for any $ i,j\in \Lambda $. For future use we also note that in our setting the Wegner estimate 
\begin{equation}
\label{eq:Wegner} \pr{\norm{G_\Lambda^E}\ge T}\lesssim A_0 |\Lambda|/T, 
\end{equation}
follows, for example, from \cite[(2.4)]{CGK-09-Generalized}.
\begin{proposition}
\label{prop:analysis-determinant} Let $ \Lambda_0=\{n\}_W\subset\Lambda=[a,b]_W $. For any
\begin{equation*}
	T\ge 
	\max(|E|,\norm{S})
\end{equation*}
 there exists a set $ \mc{B}=\mc{B}(n,T)\subset \R^{\Lambda_0'} $, with $ 
\pr{\mc{B}}\lesssim WA_0/T $, such that
\begin{equation*}
	f_\Lambda^E(V,V')\neq 0
\end{equation*}
 for any $ V\in \C^{\Lambda_0} 
$, $ \min_{i\in \Lambda_0}|V_i|\ge 10WT $, $ V'\in \R^{\Lambda_0'}\setminus \mc{B} $. 
\end{proposition}
\begin{proof}
Using \eqref{eq:apx-H-Lambda-direct-sum} and Lemma~\ref{lem:apx-Schur-complement} we have 
\begin{equation*}
	f_\Lambda^E=\det(H_\Lambda/H_{\Lambda_0'}-E) \det(H_{\Lambda_0'}-E),
\end{equation*}
 where 
\begin{equation}
	H_\Lambda/H_{\Lambda_0'}=H_{\Lambda_0}-\Gamma_0 G_{\Lambda_0'}^E \Gamma_0^* 
	=\diag(V_{(n,1)},\ldots,V_{(n,W)})+S-\Gamma_0 G_{\Lambda_0'}^E \Gamma_0^*. 
\end{equation}
If $ |G_{\Lambda_0'}^E(k,l)|\le T $ for any $ k,l\in 
\partial_\Lambda\Lambda_0$ then $ |(\Gamma_0 G_{\Lambda_0'}^E \Gamma_0^*)(i,j)|\le 4T $ for any $ 
i,j\in\Lambda_0 $, and consequently $ \norm{\Gamma_0 G_{\Lambda_0'}^E \Gamma_0^*}\le 4WT $. 
Furthermore, if we also have that $ \min_{i\in \Lambda_0}|V_i|\ge 10WT $ and $ 
T\ge\max(|E|,\norm{S}) $, then $ H_\Lambda/H_{\Lambda_0'}-E $ is invertible since 
\begin{equation*}
	\norm{\diag(V_{(n,1)},\ldots,V_{(n,W)})^{-1}}\cdot \norm{-E+S-\Gamma_0 G_{\Lambda_0'}^E 
	\Gamma_0^*} \le \frac{6WT}{10WT}<1. 
\end{equation*}
The conclusion follows by setting 
\begin{equation*}
	\mc{B} =\{V'\in \R^{\Lambda_0'}:\,|G_{\Lambda_0'}^E(k,l)|> T,\, k,l\in
	\partial_\Lambda\Lambda_0\} \cup \{V'\in \R^{\Lambda_0'}:\, \det(H_{\Lambda_0'}-E)=0\}. 
\end{equation*}
The bound on $ \pr{\mc{B}} $ follows from \eqref{eq:Wegner-for-entries}. 
\end{proof}

We can now prove Theorem~\ref{thm:intro-var-log-Green} 
\begin{proof}
[Proof of Theorem~\ref{thm:intro-var-log-Green}] The result follows by applying 
Theorem~\ref{thm:intro-var-log-rational} with 
\begin{equation*}
	P(V)=\sum|[g_\lambda^E(i,j)](V)|^2,\,Q(V)=|f_\Lambda^E(V)|^2,\,J_k=\{k\}_W,\,k\in(a,b).  
\end{equation*}
Note that $ P $ and $ Q $ are 
polynomials of real variables, but with possibly complex coefficients. The assumptions on $ P $ and 
$ Q $ are satisfied due to Proposition~\ref{prop:analysis-minors} and 
Proposition~\ref{prop:analysis-determinant}. 
\end{proof}

To establish the bounds on the moments we need the following Cartan's estimate for Green's function.
\begin{lemma}
\label{lem:analysis-Cartan-R} There exist absolute constants $ C_0 $ and $ C_1 $ such that for any 
$ R\ge e $ and $ H\gg 1 $ we have 
\begin{equation*}
	\mes\left\{V\in\R^\Lambda:~\norm{V}\le R,~ \log \SigmaLE \le -C_0H M_R \right\}\le 
	C_1^{|\Lambda|}R^{|\Lambda|}\exp(-H), 
\end{equation*}
where $ M_R=|\Lambda|\max(1,\log|E|,\log\norm{S})\log R $. 
\end{lemma}
\begin{proof}
We have 
\begin{equation*}
	\norm{H_\Lambda^{ij}(V)-E}\le 1+\norm{H_\Lambda(V)-E}\le 1+|E|+R+\norm{S}, 
\end{equation*}
for any $ V\in\C^\Lambda $, $ \norm{V}\le R $, and any $ i,j\in\Lambda $ (recall that $ 
H_\Lambda^{ij} $ was defined in the proof of Proposition~\ref{prop:analysis-minors}). Consequently, 
there exists an absolute constant $ B $ such that 
\begin{equation}
	\sup_{\norm{V}\le R}\log|f_\Lambda^E(V)|
	\le |\Lambda| \log(|E|+R+\norm{S})
	\le B |\Lambda| \max(1,\log|E|,\log\norm{S}) \log R 
\end{equation}
and 
\begin{equation}
	\sup_{\norm{V}\le R}\log|[g_\Lambda^E(i,j)](V)|\le |\Lambda| \log(1+|E|+R+\norm{S}) 
	\le B |\Lambda| \max(1,\log|E|,\log\norm{S}) \log R, 
\end{equation}
for $ R\ge e $. Let 
\begin{equation*}
	M=B |\Lambda| \max(1,\log|E|,\log\norm{S}) \log R
\end{equation*}
 and $ C_0 $ as in 
Lemma~\ref{lem:apx-logP-Cartan}. If 
\begin{equation*}
	\log \SigmaLE \le -3C_0HM 
\end{equation*}
then 
\begin{equation*} 
 	\log|[g_\Lambda^E(i',j')]| \le \frac{1}{2} (\log \SigmaLE+\log|f_\Lambda^E|) \le 
 	-\frac{3}{2}C_0 H M + 
 	\frac{1}{2}\log|f_\Lambda^E|\le -C_0 H M,
\end{equation*}
where we chose $ i'\in\{a\}_W,j'\in\{b\}_W $ (assuming $ \Lambda=[a,b]_W $) such that $ i'_2=j'_2 
$. The conclusion follows by applying Lemma~\ref{lem:apx-logP-Cartan} to $ 
\log|[g_\Lambda^E(i',j')]| $. This is possible due to Proposition~\ref{prop:analysis-minors} (ii). 
Note that the constant $ C_0 $ from the result is not the same as in 
Lemma~\ref{lem:apx-logP-Cartan}. 
\end{proof}
\begin{proposition}
\label{prop:analysis-moments} Given $ s\ge 1 $ there exists a constant 
\begin{equation*}
	C_0=C_0(A_0,A_1,|E|,s,\norm{S})
\end{equation*}
 such that 
\begin{equation*}
	\ex{\log^s\SigmaLE}\le C_0 (|\Lambda|\log|\Lambda|)^{2s},\, |\Lambda|>1. 
\end{equation*}
\end{proposition}
\begin{proof}
From Lemma~\ref{lem:analysis-Cartan-R} and Lemma~\ref{lem:apx-Cartan->integrability} it follows 
that for any $ R\ge e $ we have 
\begin{equation*}
	\int_{\norm{V}\le R}\log^s\SigmaLE \,d\nu \le \left( C |\Lambda|^{2}\log^2 R \right)^s, 
\end{equation*}
with $ C=C(A_0,|E|,\norm{S}) $.

Note that due to \eqref{eq:intro-densitydecay} we have 
\begin{equation*}
	\pr{\norm{V}\ge R}\le \sum_{i\in \Lambda} \pr{|V_i|\ge R/{|\Lambda|}^{1/2}} \le A_1 
	|\Lambda|^{3/2}/R. 
\end{equation*}

Let $ R_k=R_0^k|\Lambda|^{3/2} $, with $ R_0\gg e $. Using the two previous estimates we have 
\begin{multline*}
	\ex{\log^s \SigmaLE} = \int_{ \norm{V} \le R_1 } \log^s \SigmaLE \,d\nu + \sum_{k=1}^{\infty} 
	\int_{R_k < \norm{V} \le R_{k+1}} \log^s \SigmaLE \,d \nu\\
	\le (C|\Lambda|^{2}\log^2R_1)^s + \sum_{k=1}^{\infty} \left(\int_{\norm{x}\le 
	R_{k+1}}\log^{2s}\SigmaLE \,d\nu\right)^{1/2} \left(\pr{\norm{V}\ge R_k}\right)^{1/2}\\
	\le (C|\Lambda|\log|\Lambda|)^{2s}
		+(C|\Lambda|\log|\Lambda|)^{2s}\sum_{k=1}^{\infty}(\log^2 R_0^{k+1})^s (A_1/R_0^k)^{1/2}\\
	\le C(s) (|\Lambda|\log|\Lambda|)^{2s}. 
\end{multline*}
\end{proof}

\section{Large Fluctuations Imply Exponential Decay}\label{sec:decay}

In this section we show how to pass from fluctuations of the resolvent to exponential decay. The 
main result is Theorem~\ref{thm:decay}. The basic idea, developed in 
Proposition~\ref{prop:decay-var->weak-decay}, is that having some fluctuations of Green's function 
implies some exponential decay with non-zero probability. The desired result will follow by 
standard multi-scale analysis. The initial estimate is provided in 
Proposition~\ref{prop:decay-initial-length} and the inductive step is implemented in 
Proposition~\ref{prop:decay-induction-step} (cf. \cite[Lemma 4.1]{vDK-89-new}). Throughout this 
section we assume 
\begin{equation*}
\var\left(\log \SigmaLE \right) \ge L \delta_0, 
\end{equation*}
with $ \delta_0 \le 1/W $, for any $ \Lambda=[a,b]_W $, $ b-a+1=L $.
\begin{proposition}
\label{prop:decay-var->weak-decay} Given $ \epsilon\in(0,1) $ there exists $ 
C_0=C_0(A_0,A_1,\epsilon,|E|,\norm{S}) $ such that 
\begin{equation*}
	\pr{\log\SigmaLE \le -\sqrt{L\delta_0}/2}\ge \left(\frac{L \delta_0}{C_0 
	|\Lambda|^4\log^4|\Lambda|}\right)^{1+\epsilon}, 
\end{equation*}
for any $ \Lambda=[a,b]_W $, $ b-a+1=L \ge C_0 \delta_0^{-1}\log^2\delta_0 $. 
\end{proposition}
\begin{proof}
We partition $ \R^\Lambda $ by the sets 
\begin{align*}
	\Omega_{-1} &=\{V:~\log\SigmaLE \le -\sqrt{L\delta_0}/2\},\\
	\Omega_{0} &=\{V:~|\log\SigmaLE| < \sqrt{L\delta_0}/2 \},\\
	\Omega_{1} &=\{V:~\log\SigmaLE \ge \sqrt{L\delta_0}/2\}.\\
\end{align*}
By our assumption on the variance we have that $ \ex{\log^2 \SigmaLE} \ge L \delta_0 $. At the same 
time we have $ \int_{\Omega_0} \log^2 \SigmaLE \,d\nu \le L\delta_0/4 $ and 
\begin{multline*}
	\int_{\Omega_{-1}} \log^2 \SigmaLE \,d\nu \le \left(\int_{\R^\Lambda} 
	\log^{2(1+\epsilon)/\epsilon} \SigmaLE \,d\nu \right)^{\epsilon/(1+\epsilon)} \left(\pr{V\in 
	\Omega_{-1}}\right)^{1/(1+\epsilon)}\\
	\le C |\Lambda|^4\log^4|\Lambda| \left(\pr{V\in \Omega_{-1}}\right)^{1/(1+\epsilon)}, 
\end{multline*}
\begin{equation*}
	\int_{\Omega_1} \log^2 \SigmaLE \,d\nu 
	\le \left(\int_{\R^\Lambda} \log^4 \SigmaLE \,d\nu 
		\right)^{1/2} \left(\pr{V\in \Omega_1}\right)^{1/2} 
	\le C |\Lambda|^4\log^4|\Lambda| 
		\left(\pr{V\in \Omega_1}\right)^{1/2}, 
\end{equation*}
with $ C=C(A_0,A_1,\epsilon,|E|,\norm{S}) $, due to Proposition~\ref{prop:analysis-moments}. We 
conclude 
\begin{equation*}
	\pr{\log \SigmaLE \le -\sqrt{L\delta_0}/2} 
	\ge \left(\frac{3 L\delta_0/4 - C|\Lambda|^4\log^4|\Lambda| \left( 
		\pr{\log \SigmaLE \ge \sqrt{L\delta_0}/2 
		\right)^{1/2}}}{C|\Lambda|^4\log^4|\Lambda|}\right)^{1+\epsilon}. 
\end{equation*}
Now we just need to estimate the probability on the right-hand side. If $ \log \SigmaLE \ge 
\sqrt{L\delta_0}/2 $ then $ |G_\Lambda^E(i,j)| \ge \exp(\sqrt{L\delta_0}/2)/W^2 $ for some $ 
(i,j)\in 
\partial \Lambda $, $ i_1<j_1 $. Using the estimate \eqref{eq:Wegner-for-entries} we have 
\begin{equation*}
	\pr{\log \SigmaLE \ge \sqrt{L\delta_0}/2}\lesssim A_0 W^4 
	\exp(-\sqrt{L\delta_0}/2). 
\end{equation*}
The conclusion follows because 
\begin{equation*}
	3L\delta_0/4 - C|\Lambda|^4\log^4|\Lambda| \left(A_0 W^4 
	\exp(-\sqrt{L\delta_0}/2) 
	\right)^{1/2} \ge L \delta_0/4, 
\end{equation*}
for $ L\ge C'\delta_0^{-1}\log^2\delta_0 $ (recall that we are assuming $ \delta_0 \le W^{-1} $). 
\end{proof}
\begin{proposition}
\label{prop:decay-initial-length} Fix $ \beta\ge 1 $. There exists $ 
C_0=C_0(A_0,A_1,\beta,|E|,\norm{S}) $ such that 
\begin{equation*}
	\pr{\log|G_{\Lambda_L(a)}^E(i,j)|\le-\delta_0^{1/2}L^{1/10}/4,\,i\in\{a\}_W, j\in
	\partial \Lambda_L(a)} \ge 1-L^{-\beta}, 
\end{equation*}
for any $ L\ge C_0 \delta_0^{-6} W^{20} $. 
\end{proposition}
\begin{proof}
We only prove that 
\begin{equation*}
	\pr{\log|G_{\Lambda_L(a)}^E(i,j)|\le-\delta_0^{1/2}L^{1/10}/4,\,i\in\{a-L\}_W, j\in\{a\}_W} \ge 
	1-L^{-\beta}/2. 
\end{equation*}
The same estimate with $ i\in\{a\}_W $ and $ j\in \{a+L\}_W $ will hold by an analogous proof.

Let $ l=[L^{1/5}] $. We have $ l^5\le L < 2 l^5 $ (provided $ L $ is larger than some absolute 
constant). Let $ \mc G_1 $ be the event that $ \log \SigmaLE[0] \le -\sqrt{l\delta_0}/2 $ holds for 
at least one block 
\begin{equation*}
	\Lambda_0= [nl+1,(n+1)l]_W \subset \Lambda=[a-L,a]_W. 
\end{equation*}
Clearly $ \Lambda $ contains more than $ l^4/2 $ such blocks. By the independence of the potentials 
and by Proposition~\ref{prop:decay-var->weak-decay} we have that for $ \epsilon $ small enough 
\begin{multline*}
	\pr{\R^\Lambda\setminus \mc G_1}\le (1-c(\delta_0 
	l)^{1+\epsilon}/(lW)^{4(1+2\epsilon)})^{l^4/2} \\ \le 
	\exp\left(-c(\delta_0 
		l)^{1+\epsilon}/(lW)^{4(1+2\epsilon)} l^4\right)
	\le \exp\left(-c\delta_0^{1+\epsilon} W^{-4(1+2\epsilon)} L^{(1-7\epsilon)/5}\right)\le 
	L^{-\beta}/4, 
\end{multline*}
provided that $ L\ge C \delta_0^{-6}W^{20} $. Let $ \mc G_2 $ be the event that $ 
\norm{G_{\Lambda_L(a)}^E}\le T $ and $ \norm{G_{\Lambda_1}^E} \le T$ for any 
\begin{equation*}
	\Lambda_1=[a-L,(n+1)l]_W\subset \Lambda, 
\end{equation*}
with $ T\ge 1 $ to be chosen later. From \eqref{eq:Wegner} it follows that 
\begin{equation*}
	\pr{\R^\Lambda\setminus \mc G_2}\lesssim A_0 L^2 W T^{-1}. 
\end{equation*}

For the event $ \mc G_1 \cap \mc G_2 $ it follows, by using the second resolvent identity 
\eqref{eq:apx-2nd-resolvent-identity}, that 
\begin{multline*}
	|G_{\Lambda_L(a)}^E(i,j)| =\left|\sum_{(k,k') \in 
	\partial_{\Lambda_L(a)} \Lambda_1} G_{\Lambda_1}^E(i,k) G_{\Lambda_L(a)}^E(k',j)\right| \le T W 
	|G_{\Lambda_1}^E(i,\tilde k)|\\
	=TW\left|\sum_{(l,l') \in 
	\partial_{\Lambda_1} \Lambda_0} G_{\Lambda_0}^E(\tilde k,l) G_{\Lambda_1}^E(l',i) \right| \le T W 
	\exp(-\sqrt{l\delta_0}/4)|G_{\Lambda_1}^E(\tilde l,i)|\\
	\le T^2 W \exp(-\sqrt{l\delta_0}/4)\le \exp(-\delta_0^{1/2}L^{1/10}/8), 
\end{multline*}
provided $ T=\exp(\delta_0^{1/2}L^{1/10}/16) $ and $ L\ge C\delta_0^{-5}\log^{10}W $. The 
conclusion follows by noticing that with this choice of $ T $ we have 
\begin{equation*}
	A_0L^2W T^{-1}\le L^{-\beta}/4, 
\end{equation*}
for $ L\ge C \delta_0^{-5}\log^{10}W $. 
\end{proof}
\begin{proposition}
\label{prop:decay-induction-step} Fix $ \beta\ge 1 $ and $ \epsilon\in(0,1) $. There exists a 
constant $ C_0=C_0(\beta,\epsilon,A_0) $ such that if for some $ l\ge C_0 $ we have 
\begin{equation*}
	\pr{\log|G_{\Lambda_l(a)}^E(i,j)|\le-m_l l,\,i\in\{a\}_W, j\in
	\partial \Lambda_l(a)} \ge 1-l^{-\beta}, 
\end{equation*}
with $ m_l\ge l^{\epsilon-1} \log W $, for any $ \Lambda_l(a)\subset \Z_W $, then for $ L=l^\alpha 
$, $ \alpha\in[2,4] $, and any $ \Lambda_L(a)\subset \Z_W $ we have 
\begin{equation*}
	\pr{\log|G_{\Lambda_L(a)}^E(i,j)|\le-m_L L,\,i\in\{a\}_W, j\in
	\partial \Lambda_L(a)} \ge 1-L^{-\beta}, 
\end{equation*}
with 
\begin{equation*}
	m_l\ge m_L\ge (1-6 l^{-1/4})m_l-\log(2W)/l\ge L^{\epsilon-1}\log W. 
\end{equation*}
\end{proposition}
\begin{proof}
Let $ I=[a-L+l,a+L-l] $. We say that $ b\in I $ is good if 
\begin{equation*}
	\log|G_{\Lambda_l(b)}^E(i,j)|\le-m_l l,\,i\in\{b\}_W, j\in
	\partial \Lambda_l(b). 
\end{equation*}
We partition $ I $ into $ 2l+1 $ subsets $ I_s=\{b\in I:~b=s~(\mod 2l+1)\} $. For each $ s $ the 
set $ I_s $ has at least $ n= (2L-4l+1)/(2l+1)-1 $ elements and the blocks $ \Lambda_l(b) $, $ b\in 
I_s $ are disjoint. By Hoeffding's inequality (see \cite[Theorem 1]{Hoe-63-Probability}) applied to 
the binomial distribution with parameters $ n $ and $ p=1-l^{-\beta} $ we have that there exist at 
least $ (1-\delta) p n $ good $ b $'s in $ I_s $, with probability greater than $ 
1-\exp(-2(pn-(1-\delta)pn)^2/n) $. Let $ B $ be the number of bad $ u\in I $. By choosing $ 
\delta=l^{-1/4} $ it follows that 
\begin{equation*}
	B\le 2L-2l+1-(2l+1)(1-\delta)pn
	=(2L-2l+1)[1-(1-\delta)p]+(4l+1)(1-\delta)p\le 4L l^{-1/4}, 
\end{equation*}
with probability greater than 
\begin{equation*}
	1-(2l+1)\exp(-2np^2\delta^2)\ge 1-(2l+1)\exp(-c L \delta^2/l) \ge 1-(2l+1)\exp(-cl^{1/2}) \ge 
	1-L^{-\beta}/2, 
\end{equation*}
provided that $ l\ge C=C(\beta) $.

Let $ \Lambda_t $ be the blocks corresponding to the connected components of the set of bad 
elements in $ I $. Clearly $ t\le B $ and if $ l_t $ is the length of $ \Lambda_t $ then $ \sum 
l_t=B $. Using \eqref{eq:Wegner} we know that with probability greater than $ 1-CA_0WL^3 T^{-1} $ 
we have $ \norm{G_\Lambda^E}\le T $, where $ \Lambda $ is any of the blocks $ \Lambda_t $ or $ 
\Lambda_L(a) $. We will choose $ T $ later.

Let $ i\in\{a\}_W $ and $ j\in 
\partial \Lambda_L(a) $. We will use the resolvent identity \eqref{eq:apx-2nd-resolvent-identity}. 
If $ a $ is good then 
\begin{equation*}
	|G_{\Lambda_L(a)}^E(i,j)| =\left| \sum_{(k,k')\in 
	\partial_{\Lambda_L(a)}\Lambda_l(a)} G_{\Lambda_l(a)}^E(i,k)G_{\Lambda_L(a)}^E(k',j) \right| 
	\le 2W \exp(-m_l l) |G_{\Lambda_L(a)}^E(\tilde k,j)|, 
\end{equation*}
for some $ \tilde k \in 
\partial_{\Lambda_L(a)} \Lambda_l(a) $. If $ a $ is bad then $ \{a\}_W \subset \Lambda_t $ and by 
our choice of $ \Lambda_t $ we know that $\tilde k_1 $ is good for any $ \tilde k\in 
\partial_{\Lambda_L(a)} \Lambda_t $ (provided $ k_1 \in I $). So if $ a $ is bad we have 
\begin{multline*}
	|G_{\Lambda_L(a)}^E(i,j)| =\left| \sum_{(k,k')\in 
	\partial_{\Lambda_L(a)}\Lambda_t} G_{\Lambda_t}^E(i,k)G_{\Lambda_L(a)}^E(k',j) \right| \le 2W T 
	|G_{\Lambda_L(a)}^E(\tilde k,j)|\\
	= 2 W T \left| \sum_{(l,l')\in 
	\partial_{\Lambda_L(a)}\Lambda_l(\tilde k_1)} G_{\Lambda_l(\tilde k_1)}^E(\tilde 
	k,l)G_{\Lambda_L(a)}^E(l',j) \right|\\
	\le 4W^2 T \exp(-m_l l) |G_{\Lambda_L(a)}^E(\tilde l,j)| = |G_{\Lambda_L(a)}^E(\tilde l,j)|, 
\end{multline*}
where we chose $ T=\exp(m_l l)/(4W^2) $. We can iterate these estimates as long as $ \tilde k_1, 
\tilde j_1\in I $. We conclude that 
\begin{equation*}
	|G_{\Lambda_L(a)}^E(i,j)|\le T (2W \exp(-m_l l))^{n_1}\le (2W \exp(-m_l l))^{n_1-2}, 
\end{equation*}
with $ n_1\ge (L-l+1-B)/(l+1)-1 $. So we have 
\begin{equation*}
	m_L=\frac{n_1-2}{L}(m_l l - \log(2W)) \ge \frac{1-5l^{-1/4}}{l+1}(m_l l -\log(2W)) \ge 
	(1-6l^{-1/4})m_l-\log(2W)/l, 
\end{equation*}
for $ l\ge C $. The conclusion follows by noting that 
\begin{equation*}
	1-CA_0WL^3 T^{-1}= 1-CA_0W^3L^3\exp(-m_l l) \ge 1-CA_0W^3L^3\exp(-l^{\epsilon} \log W ) \ge 
	1-L^{-\beta}/2, 
\end{equation*}
provided $ l\ge C=C(\beta,\epsilon,A_0) $. 
\end{proof}
\begin{theorem}
\label{thm:decay} Fix $ \beta\ge 1 $. If $ \var(\SigmaLE)\ge L\delta_0 $, with $ \delta_0\le W^{-1} 
$, for any $ \Lambda=[a,b]_W $, with $ b-a+1=L $, then there exists $ 
C_0=C_0(A_0,A_1,\beta,|E|,\norm{S}) $ such that 
\begin{equation*}
	\pr{\log|G_{\Lambda_L(a)}^E(i,j)|\le-C_0^{-1} \delta_0^{6} W^{-20}L,\,i\in\{a\}_W, j\in
	\partial \Lambda_L(a)} \ge 1-L^{-\beta}, 
\end{equation*}
for any $ L\ge C_0 \delta_0^{-12} W^{40} $ and $ a\in\Z $. 
\end{theorem}
\begin{proof}
Let $ L_0=B \delta_0^{-6} W^{20} $. If $ B $ is large enough, as in 
Proposition~\ref{prop:decay-initial-length}, then 
\begin{equation*}
	\pr{\log|G_{\Lambda_{L_0}(a)}^E(i,j)|\le- m_{L_0}L_0,\,i\in\{a\}_W, j\in
	\partial \Lambda_{L_0}(a)} \ge 1-L_0^{-\beta}, 
\end{equation*}
with $ m_{L_0}=\delta_0^{1/2}L_0^{1/10}/(4L_0)=B^{-9/10}\delta_0^{59/10} W^{18}/4 $. Note that 
\begin{equation*}
	m_{L_0}\ge L_0^{1/100-1} \log W
\end{equation*}
 provided $ B $ is large enough.

Given $ L\ge L_0^2 $ we can find a sequence $ {L_k} $ such that $ L_{k+1}=L_k^{\alpha_k} $, $ 
\alpha_k\in[2,4] $ and $ L=L_{k_0} $ for some $ k_0\ge 1 $. Applying 
Proposition~\ref{prop:decay-induction-step} inductively we have 
\begin{equation*}
	m_{L_{k+1}}\ge (1-L_k^{-1/4})m_{L_k}-\log(2W)/L_k. 
\end{equation*}
Consequently we get 
\begin{equation*}
	m_L-m_{L_0}\ge -\sum_{k=0}^{\infty}\left(m_{L_k}L_k^{-1/4}+\log (2W)L_k^{-1}\right) \ge 
	-m_{L_0}/2, 
\end{equation*}
provided that $ B $ is large enough (we used the fact that $ m_{L_0}\ge m_{L_k} $ and $ 
m_{L_0}\ge L_0^{1/100-1} \log W $). The conclusion follows immediately. 
\end{proof}
\appendix

\section{Cartan's Estimate}

For convenience we include a statement of the Cartan estimate for analytic functions (see 
\cite[Theorem 11.4]{Lev-96-Lectures}).

\begin{lemma}
	Let $ \phi:\mathbb{D}\to \C $ be an analytic function such that
	\begin{equation*}
		m\le \log|\phi(0)|,\,M\ge \sup_{\zeta\in 
			\mathbb{D}}\log|\phi(\zeta)|.
	\end{equation*} 
	Then there exists an absolute constant
	$ C_0 $ such that for any $ H\gg 1 $ we have
	\begin{equation*}
		\log|\phi(\zeta)|>M-C_0H(M-m), 
	\end{equation*}
	for all $ \zeta\in \mathbb{D}_{1/6} $ except for a set of disks with the sum of the radii less 
	than $ \exp(-H) $.
\end{lemma} 
The next result is a Cartan type estimate for multivariate polynomials.
\begin{lemma}
\label{lem:apx-logP-Cartan} If $ P(x)=\sum_{|\alpha|\le D} a_\alpha x^\alpha $ is a polynomial of $ 
N $ variables such that $ \max_{|\alpha|\le D} |a_\alpha|\ge 1$ and $ \sup_{\norm{z}\le 20 R_0} 
\log|P(z)|\le M_{R_0} $, for some $ R_0\ge 1 $, then there exist absolute constants $ C_0 $ and $ 
C_1 $ such that for any $ H\gg 1 $ we have 
\begin{equation*}
	\mes \{ x\in\R^N:~\norm{x}\le R_0,\log|P(x)|\le-C_0H M_{R_0} \} \le C_1^N R_0^N\exp(-H). 
\end{equation*}
\end{lemma}
\begin{proof}
The strategy is to apply the one dimensional Cartan's estimate on complex lines that will cover 
the set $ \{ \norm{x}\le R_0 \} $. For this we need to find a point $ x_0\in \R^N $ at which $ 
|P(x_0)| $ is bounded away from zero. Due to the Cauchy estimates for the derivatives of analytic 
functions one has 
\begin{equation*}
	|a_\alpha|\le \max_{\norm{z}\le 1} |P(z)|, 
\end{equation*}
for any $ \alpha $. It follows that there exists $ z_0\in\C^N $, $ \norm{z_0}\le 1 $, such that $ 
|P(z_0)|\ge1 $. We will use Cartan's estimate ``centered'' at $ z_0 $ to show the existence of 
$ x_0 $. Let $ \phi(\zeta)=P(z_0-10\zeta\Im z_0) $. This peculiar definition is motivated by the 
fact that $ z_0-10\zeta\Im z_0\in \R^N $ whenever $ \Im \zeta=1/10 $.  We have that $ 
\log|\phi(0)|\ge 0 $ and 
$ \sup_{\zeta\in\D }\log|\phi(\zeta)|\le M_{R_0} $, so Cartan's estimate guarantees, in particular, 
that there exists  $ |\zeta_0|\le 1/6 $ with $ \Im \zeta_0=1/10 $ such that 
\begin{equation*}
	\log|\phi(\zeta_0)|\ge -CM_{R_0},
\end{equation*}
with $ C\gg 1 $. We can now choose $ x_0=z_0-10\zeta_0\Im z_0 
$.

Let $ f(z)=P(x_0+12 R_0 z) $. We have that 
\begin{equation*}
	\log|f(0)|\ge -CM_{R_0},\quad \sup_{\norm{z}\le 1} \log|f(z)|\le \sup_{\norm{z}\le 20R_0} \log 
	|P(z)|\le M_{R_0}, 
\end{equation*}
and 
\begin{multline*}
	\{ x\in\R^N:~\norm{x}\le R_0,\log|P(x)|\le-CH M_{R_0} \}\\
	\subset x_0+12R_0
	\underbrace{\{x\in \R^N:~\norm{x}\le 1/6,~ \log|f(x)| \le -C H M_{R_0}\}}_{:=\mc B}. 
\end{multline*}
Let $ \xi_0\in\{x\in \R^N: \norm{x}=1\} $. By applying Cartan's estimate to
\begin{equation*}
	\varphi(\zeta)=\log|f(\zeta \xi_0)|
\end{equation*}
we get $ \int_{\R} \mathbbm{1}_{\mc B}(rx_0)dr\le C\exp(-H) $. The conclusion now 
follows by integrating $ \mathbbm{1}_{\mc B} $ in hyper-spherical coordinates. 
\end{proof}

We also illustrate how to obtain explicit integrability estimates for functions satisfying a Cartan
type estimate. 
\begin{lemma}
\label{lem:apx-Cartan->integrability} Let $ f $ be a measurable function on $ \{x\in 
\R^N:~\norm{x}\le R_0\} $, $ R_0>0 $ such that 
\begin{equation*}
	\mes\{x\in \R^N:~\norm{x}\le R_0,~\log|f(x)| \le -C_0HM_0\}\le C_1^NR_0^N \exp(-H), 
\end{equation*}
for some $ M_0\ge \sup_{\norm{x}\le R_0}\log|f(x)| $, and some absolute constants $ C_0, C_1 $. 
Given $ s>0 $ there exists an absolute constant $ C_2 $ such that if $ \mu $ is a probability 
measure with $ d\mu\le B_0^N dm $ for some $ B_0>0 $, then 
\begin{equation}
	\int_{\norm{x}\le R_0}|\log|f(x)||^s \,d\mu(x)\le \left( C_2 M_0N\max(1,\log B_0,\log 
	R_0)\right)^s, \,s\ge 1. 
\end{equation}
\end{lemma}
\begin{proof}
\begin{multline*}
	\int_{\norm{x}\le R_0} |\log|f(x)||^s\,d\mu (x) 
	=\int_{0}^{\infty}\mu(|\log|f(x)||^s\ge\lambda,\,\norm{x}\le R_0) \,d\lambda\\
	=\int_{0}^{H_0}\mu(|\log|f(x)||^s\ge (CHM_0)^s,\norm{x}\le R_0)sC^sM_0^sH^{s-1}\,dH\\
	+\int_{H_0}^{\infty}\mu(\log|f(x)|\le -CHM_0,\norm{x}\le R_0)sC^sM_0^sH^{s-1}\,dH\\
	\le (CM_0H_0)^s+C^sM_0^sB_0^N\int_{H_0}^{\infty}\mes\{\log|f(x)|\le -CHM_0,\norm{x}\le 
	R_0\}sH^{s-1}\,dH\\
	\le (CM_0H_0)^s+C^{N+s} M_0^s B_0^NR_0^N\exp(-H_0/2) \le C^s M_0^sN^s(\max(1,\log B_0,\log 
	R_0))^s, 
\end{multline*}
Note that we chose $ H_0=CN\max(1,\log B_0,\log R_0) $. 
\end{proof}

\section{Resolvent Identities}

Recall the following fundamental facts regarding Schur's complement (see, for example, 
\cite[Theorem 1.1-2]{Zha-05-Schur}).
\begin{lemma}
\label{lem:apx-Schur-complement} Let 
\begin{equation*}
	H=
	\begin{bmatrix}
		H_0 & \Gamma_0\\
		\Gamma_1 & H_1\\
	\end{bmatrix}
	, 
\end{equation*}
where $H_0$ is a $n_0 \times n_0$ matrix and $H_1$ is an invertible $n_1\times n_1$ matrix. Let 
$H/H_1 = H_0 - \Gamma_0 H_1^{-1} \Gamma_1$. Then 
\begin{equation*}
	\det H= (\det H/H_1)(\det H_1) 
\end{equation*}
and if $ H/H_1 $ is invertible then 
\begin{equation*}
	H^{-1}=
	\begin{bmatrix}
		(H/H_1)^{-1} & -(H/H_1)^{-1}\Gamma_0 H_1^{-1}\\
		-H_1^{-1}\Gamma_1(H/H_1)^{-1} & H_1^{-1}+H_1^{-1}\Gamma_1(H/H_1)^{-1}\Gamma_0 H_1^{-1} 
	\end{bmatrix}
	. 
\end{equation*}
\end{lemma}

Next we set things up so that we can apply the previous lemma to our finite volume matrices. Let $ 
\Lambda=[a,b]\times[1,W] $ and $ \Lambda_0=[a_0,b_0]\times[1,W] $ be so that $ 
\Lambda_0\subset\Lambda $, and let $ \Lambda_0'=\Lambda\setminus\Lambda_0 $. By viewing 
$\C^\Lambda$ as $\C^{\Lambda_0} \oplus \C^{\Lambda_0'}$ one has the following matrix representation 
\begin{equation}
\label{eq:apx-H-Lambda-direct-sum} H_\Lambda = 
\begin{bmatrix}
	H_{\Lambda_0} & \Gamma_0\\
	\Gamma_0^* & H_{\Lambda_0'}
\end{bmatrix}
, 
\end{equation}
where 
\begin{equation}
\Gamma_0(i,j)=
\begin{cases}
	-1 & \text{if } |i_1-j_1|=1 \text{ and } i_2=j_2\\
	0 & \text{otherwise} 
\end{cases}
\end{equation}
(note that, implicitly, $ i\in\Lambda_0 $ and $ j\in\Lambda_0' $).

We recall the second resolvent identity (see, for example, \cite[Lemma 6.5]{Tes-09-Mathematical}) 
as used in \cite[(2.12)]{FS-83-Absence}. We have that $ H_\Lambda=H_{\Lambda_0}\oplus 
H_{\Lambda_0'} + \Gamma $, with 
\begin{equation*}
\Gamma=
\begin{bmatrix}
	0 & \Gamma_0 \\
	\Gamma_0^* & 0 
\end{bmatrix}
. 
\end{equation*}
The second resolvent identity gives us that 
\begin{equation*}
G_\Lambda^E=G_\oplus^E-G_\oplus^E\Gamma G_\Lambda^E, 
\end{equation*}
where $ G_\oplus^E=G_{\Lambda_0}^E\oplus G_{\Lambda_0'}^E $. We have that 
\begin{equation*}
\Gamma(i,j)=
\begin{cases}
	-1 & \text{if } (i,j)\in
	\partial_\Lambda \Lambda_0 \text{ or } (j,i)\in
	\partial_\Lambda \Lambda_0\\
	0 & \text{otherwise} 
\end{cases}
. 
\end{equation*}
It follows that for any $ i\in \Lambda_0 $ and $ j\in\Lambda_0' $ we have 
\begin{equation}
\label{eq:apx-2nd-resolvent-identity} G_\Lambda^E(i,j) =\sum_{(k,k')\in 
\partial_\Lambda \Lambda_0} G_{\Lambda_0}^E(i,k)G_\Lambda^E(k',j). 
\end{equation}

\bibliographystyle{amsalpha} 
\bibliography{Schroedinger}

\end{document}